\documentclass[11pt]{amsart}
\usepackage[margin=1in,letterpaper]{geometry}

\usepackage{amssymb,amsfonts,amsmath,bbm,mathrsfs,stmaryrd,wasysym}

\usepackage{xcolor}
\usepackage[colorlinks,
             linkcolor=black!75!red,
             citecolor=blue!50!black,
             pdftitle={Lattice Poisson AKSZ},
             pdfauthor={Theo Johnson-Freyd},
             pdfsubject={quantization},
             pdfkeywords={quantization,diagrams},
             pdfproducer={pdfLaTeX},
             pdfpagemode=None,
             bookmarksopen=true
             bookmarksnumbered=true]{hyperref}

\usepackage{amsthm}

\newtheorem{definition}[subsection]{Definition}
\newtheorem{proposition}[subsection]{Proposition}
\newtheorem{theorem}[subsection]{Theorem}
\newtheorem{claim}[subsection]{Assertion}

\newtheorem{corollary}[subsection]{Corollary}
\newtheorem{conjecture}[subsection]{Conjecture}

\renewcommand\qedhere{\hfill\qedsymbol}

\theoremstyle{remark}

\newtheorem{remark}[subsection]{Remark}

\usepackage{tikz}
\usetikzlibrary{arrows,decorations.pathmorphing,decorations.pathreplacing,decorations.markings,shapes.geometric,through,fit,shapes.symbols}
\tikzset{
    dot/.style={circle,draw,fill,inner sep=1.25pt},
    onearrow/.style={postaction={decorate}, decoration={markings,mark=at position .6 with {\arrow[draw,line width=1pt]{>}}}},
    inversearrow/.style={postaction={decorate}, decoration={markings,mark=at position .45 with {\arrow[draw,line width=1pt]{<}}}},
    twoarrows/.style={draw, postaction={decorate}, decoration={markings,mark=at position .35 with {\arrow[draw,line width=1pt]{>}},mark=at position .75 with {\arrow[draw,line width=1pt]{>}}}},
    twoarrowsempty/.style={postaction={decorate}, decoration={markings,mark=at position .3 with {\arrow[draw,line width=1pt]{>}},mark=at position .7 with {\arrow[draw,line width=1pt]{>}}}},
    inversetwoarrows/.style={draw, postaction={decorate}, decoration={markings,mark=at position .35 with {\arrow[draw,line width=1pt]{<}},mark=at position .7 with {\arrow[draw,line width=1pt]{<}}}},
    squiggly/.style={draw, decorate,decoration={snake,amplitude=.3mm,segment length=2mm}},
    fastsquiggly/.style={draw, decorate,decoration={snake,amplitude=.3mm,segment length=1mm}},
    inversesquiggly/.style={draw, decorate,decoration={snake,amplitude=.2mm,segment length=2mm},postaction={decorate,decoration={markings,mark=at position .45 with {\arrow[draw,line width=1pt]{<}}}}},
    degreeshift/.style={onearrow,dashed},
    octarine/.style={postaction={decorate,decoration={markings,mark=at position .6 with {\arrow[draw,line width=1pt]{>}}}}},
}

\newcommand\object[1]{\,\tikz[baseline=(basepoint)]\path[{#1}](0,-4pt) -- (0,12pt) (0,0) coordinate (basepoint);\,}

\newcommand\mult{\,\tikz[baseline=(basepoint)]{ 
    \path (0,0) coordinate (basepoint) (0,4pt) node[dot] {};
    \draw[](-6pt,-4pt) -- (0,4pt);
    \draw[](6pt,-4pt) -- (0,4pt);
    \draw[](0,4pt) -- (0,12pt);
  }\,}
\newcommand\comult{\,\tikz[baseline=(basepoint)]{ 
    \path (0,0) coordinate (basepoint) (0,4pt) node[dot] {};
    \draw[](0,-4pt) -- (0,4pt);
    \draw[](0,4pt) -- (-6pt,12pt);
    \draw[](0,4pt) -- (6pt,12pt);
  }\,}

\newcommand\define[1]{{\em #1}}
\newcommand\cat[1]{\textmd{\textsc{#1}}}

\newcommand\Ff{{\mathcal F}}

\newcommand\DD{{\mathbb D}}
\newcommand\KK{{\mathbb K}}
\newcommand\NN{{\mathbb N}}

\newcommand\RR{{\mathbb R}}
\renewcommand\SS{{\mathbb S}}

\newcommand\T{{\mathrm T}}
\renewcommand{\d}{{\mathrm d}}

\newcommand\m{\mathbf{m}}
\newcommand\n{\mathbf{n}}

\newcommand\st{{\textrm{ s.t.\ }}}

\DeclareMathOperator{\Aut}{Aut}

\DeclareMathOperator{\spec}{Spec}

\DeclareMathOperator{\Sym}{Sym}
\DeclareMathOperator{\End}{End}

\DeclareMathOperator{\diag}{diag}

\DeclareMathOperator{\Config}{Config}

\DeclareMathOperator{\hlim}{holim}
\DeclareMathOperator{\hcolim}{hocolim}
\DeclareMathOperator{\colim}{colim}

\DeclareMathOperator{\Maps}{\underline{Maps}}

\DeclareMathOperator{\DGVect}{\cat{DGVect}}

\DeclareMathOperator{\Chains}{Chains}

\DeclareMathOperator{\Homology}{H}
\renewcommand\H\Homology

\newcommand\E{\mathrm{E}}

\newcommand{\id}{\mathrm{id}}

\newcommand\dR{\mathrm{dR}}

\DeclareMathOperator{\Frob}{Frob}
\DeclareMathOperator{\invFrob}{invFrob}
\newcommand\h{\mathrm h}
\newcommand\hFrob{\h\Frob}
\newcommand\hdiFrob{\h^{\di}\Frob}
\newcommand\hprFrob{\h^{\pr}\Frob}
\newcommand{\hBDF}{\DD(\Frob_{0})}
\newcommand\shBDF\hBDF
\DeclareMathOperator{\LB}{LB}
\newcommand\sh{\mathrm{sh}}
\newcommand\shLB{\sh\LB}
\DeclareMathOperator{\invLB}{invLB}
\DeclareMathOperator{\surinvLB}{surinvLB}

\DeclareMathOperator{\Pois}{Pois}
\DeclareMathOperator{\Lie}{Lie}

\DeclareMathOperator{\qloc}{QLoc}
\newcommand\QLoc\qloc

\DeclareMathOperator{\di}{di}
\DeclareMathOperator{\pr}{pr}
\newcommand\op{\mathrm{op}}

\newcommand\cpt{\mathrm{cpt}}

\newcommand\dist{\mathrm{dist}}

\newcommand\mono{\hookrightarrow}
\newcommand\epi{\twoheadrightarrow}
\newcommand\onto\epi

\newcommand\<\langle
\renewcommand\>\rangle

\renewcommand\[{\llbracket}
\renewcommand\]{\rrbracket}

\newcommand\Circ{\ocircle}

\title{Poisson AKSZ theories and their quantizations}
\author{Theo Johnson-Freyd}

\begin{document}

\begin{abstract}
  We generalize the AKSZ construction of topological field theories to allow the target manifolds to have possibly-degenerate up-to-homotopy Poisson structures.  Classical AKSZ theories, which exist for all oriented spacetimes, are described in terms of dioperads.  The quantization problem is posed in terms of extending from dioperads to properads.  We conclude by relating the quantization problem for AKSZ theories on $\RR^{d}$ to the formality of the $\E_{d}$ operad, and conjecture a properadic description of the space of $\E_{d}$ formality quasiisomorphisms.
\end{abstract}

\maketitle

The Alexandrov--Kontsevich--Schwarz--Zaboronsky (AKSZ) construction of classical topological field theories~\cite{MR1432574} is usually understood as follows.  One chooses an oriented $d$-dimensional ``spacetime'' manifold $M$ and a ``target'' dg supermanifold $X$ with a $(d-1)$-shifted symplectic form.  Denote by $\T[1]M = \Maps(\RR^{0|1},M)$ the dg supermanifold whose functions are the de Rham complex $\bigl(\Omega^{\bullet}(M),\partial_{\dR}\bigr)$; then the orientation on $M$ along with the symplectic form on $X$ together give the infinite-dimensional dg supermanifold $\Maps(\T[1]M,X)$ of smooth maps from $\T[1]M$ to $X$ a $(-1)$-shifted symplectic form.  Finite-dimensional $(-1)$-shifted symplectic dg supermanifolds are closely related to oscillating gauged integrals by the BV--BRST picture, and so the AKSZ construction both generalizes and clarifies the problem of defining path integrals for topological field theories.
  The classical AKSZ construction has been precisely formulated within the framework of derived algebraic geometry~\cite{PTVV}.
The corresponding derived  symplectic differential geometry is under active development~\cite{CR2,CR1}, and has been used in~\cite{CMR2,CMR1} to define extended classical field theories in the cobordism framework.

However, the infinite-dimensionality of $\Maps(\T[1]M,X)$ interferes with inverting the symplectic form into a Poisson bracket, and it is the Poisson bracket that is needed to formulate the ``quantum master equation'' controlling measures for the path integral.  Indeed, the algebraic structure of the BV--BRST picture suggests axioms for observables in perturbative quantum field theory~\cite{costellogwilliam}, but the most natural such axioms 
use possibly-degenerate Poisson brackets rather than symplectic forms.  
These considerations motivated the construction in the first part of this paper: we will develop (in the case when $X$ is an infinitesimal neighborhood of a point) a Poisson version of the AKSZ construction.  Our construction is purely obstruction-theoretic, and will not require understanding dg supermanifolds.  

Classical field theories are typically described by the values of all ``tree-level Feynman diagrams,'' whereas a perturbative quantization of such a theory should make sense of Feynman diagrams with arbitrary genus.  We encode this notion of ``classical'' in terms of actions of dioperads (which are like operads but with more complicated trees), and ``quantum'' in terms of actions of properads (which allow graphs with genus);  quantization, which might be obstructed, corresponds to extending a dioperadic action to a properadic action.

 It is known~\cite{costellogwilliam,DAG} that the factorization algebras for observables of topological field theories on $\RR^{d}$ are the same as $\E_{d}$ algebras.  Our construction does not directly produce true factorization algebras at the quantum level, but does produce a ``quasilocal'' approximation of factorization algebras, and we describe in Assertion~\ref{thm.ed} a quasilocal construction of $\E_{d}$ algebras from quantizations of our AKSZ theories.  This is evidence for Conjecture~\ref{conj}, which suggests a homotopy equivalence between the space of $\E_{d}$ formality morphisms and the space of quasilocal homotopy actions on $\Chains(\RR^{d})$ of the properad controlling involutive Frobenius algebras.

\subsection{Outline} In Section~\ref{section.diproperads} we review the notions of \define{dioperad} and \define{properad} and introduce our main characters, the (di/pr)operads controlling shifted-Frobenius algebras and shifted-Poisson infinitesimal manifolds, both up to homotopy.  In Section~\ref{section.frobchains} we construct a canonical contractible space of up-to-homotopy open shifted-Frobenius algebra structures on the chains on a smooth oriented manifold.  In Section~\ref{section.aksz}, we use these ideas to construct classical AKSZ theories with target an up-to-homotopy shifted-Poisson infinitesimal manifold, define their ``path integral quantizations,'' and conjecture that the existence of a \emph{properadic} up-to-homotopy shifted-Frobenius action on $\Chains_{\bullet}(\RR^{d})$ is equivalent to the formality of the $\E_{d}$ operad.

\subsection{Acknowledgements}
I would like to thank C.\ Arias Abad, A.\ Cattaneo, K.\ Costello, G.\ Drummond-Cole, O.\ Gwilliam, J.\ Hirsh, P.\ Mn\"ev, N.\ Reshetikhin, N.\ Rozenblyum, B.\ Vallette, A.\ Vaintrob, and T.\ Willwacher for their conversations related to this article.
I am also grateful to the anonymous referee, who provided particularly valuable comments.
Stages of this project were supported by  the NSF grants DMS-0901431, DMS-1201391, and DMS-1304054.

\section{Dioperads, properads,  shifted Frobenius, and shifted Poisson} \label{section.diproperads}

Dioperads and properads both provide axioms for algebraic structures that involve multilinear maps.  Dioperads, introduced in~\cite{MR1960128}, use only tree-level compositions; properads, from~\cite{MR2320654}, allow compositions for connected directed acyclic graphs.  We will recall the basic theory of both simultaneously.

Fix a field $\KK$ of characteristic $0$.  We will use homological conventions: a \define{chain complex} is a formal direct sum $V = \bigoplus V_{n}$ of vector spaces over $\KK$ with a differential $\partial : V_{n} \to V_{n-1}$ squaring to zero.  The tensor product and sign conventions are the usual ones for chain complexes.  Shifts are implemented by tensoring: $V[n] = V\otimes [n]$, where $\dim\,[n]_{m} = \delta_{n,m}$.  The category of chain complexes is denoted $\DGVect$.

\begin{definition}
  A \define{directed graph} consists of a finite set of \define{(internal) vertices} $v = (\m_{v},\n_{v})$, where $\m_{v}$ is a finite set of \define{arriving edges} and $\n_{v}$ is a finite set of \define{departing edges},  finite sets $\m_{\Gamma}$ and $\n_{\Gamma}$ of \define{incoming} and \define{outgoing external edges}, and a bijection between $\n_{\Gamma} \sqcup \bigsqcup_{v} \m_{v}$ and $\m_{\Gamma}\sqcup \bigsqcup_{v}\n_{v}$, where $\sqcup$ denotes disjoint union.  A \define{dag}, short for \define{connected directed acyclic graph}, is a directed graph with one component and no oriented cycles.  The \define{genus} of a dag $\Gamma$ is its first Betti number~$\beta(\Gamma)$.  A dag $\Gamma$ is a \define{directed tree} if $\beta(\Gamma) = 0$.  We denote the size of the set $\m_{v}$ by $m_{v} = |\m_{v}|$, and similarly for $\n_{v},\m_{\Gamma},\n_{\Gamma}$.  When we draw dags, we will adopt the convention that edges are always oriented from bottom to top.
    
  Let $\SS$ denote the groupoid of finite sets and bijections.  An \define{$\SS$-bimodule} is a functor $\SS^{\op} \times \SS \to \DGVect$.
  A \define{dioperad (resp.\ properad)} is an $\SS$-bimodule $P$ along with, for each directed tree (resp.\ dag) $\Gamma$ a \define{composition} map:
  $$ \Circ_{\Gamma} : \bigotimes_{v\in \Gamma} P(\m_{v},\n_{v}) \to P(\m_{\Gamma},\n_{\Gamma}) $$
  These compositions maps must be compatible with the $\SS$-bimodule structures in the obvious way.  Moreover, we demand an associativity axiom: if a graph $\Gamma$ is formed from a larger graph $\Gamma'$ by contracting a subgraph $\Gamma''$ to a vertex, then $\Circ_{\Gamma} = \Circ_{\Gamma'} \circ \Circ_{\Gamma''}$.
  
  A properad $P$ is \define{genus-graded} if each $P(\m,\n)$ decomposes as $P(\m,\n) = \bigoplus_{\beta \in \NN} P(\m,\n,\beta)$ and for each dag $\Gamma$ the composition $\Circ_{\Gamma}$ has weight $\beta(\Gamma)$.  A genus-graded properad $P$ is \define{positive} if $P(\m,\emptyset,0) = P(\emptyset,\n,0) = 0$ for all $\m,\n$ and $P(\mathbf 1,\mathbf 1,0) = \KK \,\object{onearrow,draw}\,$, where $\mathbf 1 = \{*\}$ is a set of size $1$, and $\,\object{onearrow,draw}\, = \Circ_{\,\tikz[baseline=(basepoint)]\path[onearrow,draw](0,-2pt) -- (0,6pt) (0,0) coordinate (basepoint);\,}(1)$ denotes the unit for composition.  A genus-graded properad $P$ is \define{locally finite-dimensional} if each $P(\m,\n,\beta)$ is a bounded finite-dimensional chain complex.  A dioperad $P$ is \define{positive} if $P(\m,\emptyset) = P(\emptyset,\n) = 0$ and $P(\mathbf 1,\mathbf 1) = \KK$, and \define{locally finite-dimensional} if each $P(\m,\n)$ is bounded finite-dimensional.
\end{definition}
  
  The categories of dioperads, properads, and genus-graded properads each have  model category structures in which the fibrations and weak equivalences are precisely the fibrations and weak equivalences of underlying $\SS$-bimodules, i.e.\ the fibrations are the surjective homomorphisms, and the weak equivalences are the homomorphisms that are isomorphisms on homology.  The cofibrations are defined by the left lifting property.  See~\cite[Appendix~A]{MR2572248} for details.

The \define{endomorphism} (di/pr)operad of a chain complex $V$ is $\End(V)(\m,\n) = \hom(V^{\otimes \m},V^{\otimes\n})$.  An \define{action} of a (di/pr)operad $P$ on a chain complex $V$ is a homomorphism $P \to \End(V)$; $V$ is then a \define{$P$-algebra}.
There is a forgetful functor from properads to dioperads; its left adjoint freely generates the \define{universal enveloping properad} from a dioperad.  
  For example, $\operatorname{Forget}$ takes endomorphism properads to endomorphism dioperads, and so actions of a dioperad are the same as actions of its universal enveloping properad.  
  An action of $P$ \define{up to homotopy} is an action of any cofibrant replacement of $P$; model category theory ensures that (up to homotopy equivalence) the choice of cofibrant replacement doesn't matter.

\begin{remark}
It will be technically convenient, although not strictly necessary, to restrict our attention just to actions of positive locally finite-dimensional (di/pr)operads, which are automatically augmented.  To drop this restriction causes no conceptual changes, but makes formulas more complicated: in Proposition~\ref{prop.koszulity}, one must use the theory of ``curved'' Koszul duality developed by~\cite{MR2993002}; in Assertion~\ref{thm.ed}, one must take more care with certain limits.  We will use properads to study, among other things, Poisson structures on infinitesimal manifolds, and positive properads can only describe Poisson structures that vanish at the origin.  This is no real loss of generality: one may always multiply a Poisson structure by a formal ``coupling constant,'' interpreted as the addition of one extra dimension to the manifold, and thereby construct a Poisson manifold with Poisson structure vanishing at the origin.
\end{remark}

\begin{definition}
  The (di/pr)operad $\Frob_{d}$ of \define{$d$-shifted commutative open and coopen Frobenius algebras} has generators
  $$\underbrace{
  \,\tikz[baseline=(basepoint)]{ 
    \path (0,0) coordinate (basepoint) (0,4pt) node[dot] {};
    \draw[onearrow](0,-6pt) -- (0,4pt);
    \draw[onearrow](0,4pt) -- (-8pt,14pt);
    \draw[onearrow](0,4pt) -- (8pt,14pt);
  }\,
  = 
  \,\tikz[baseline=(basepoint)]{ 
    \path (0,0) coordinate (basepoint) (0,4pt) node[dot] {};
    \draw[onearrow](0,-6pt) -- (0,4pt);
    \draw[](0,4pt) .. controls +(8pt,4pt) and +(8pt,-4pt) .. (-8pt,14pt);
    \draw[](0,4pt) .. controls +(-8pt,4pt) and +(-8pt,-4pt) .. (8pt,14pt);
  }\,
  ,}_{\text{homological degree $0$}} \quad 
  \underbrace{\,\tikz[baseline=(basepoint)]{ 
    \path (0,0) coordinate (basepoint) (0,4pt) node[dot] {};
    \draw[onearrow](-8pt,-6pt) -- (0,4pt);
    \draw[onearrow](8pt,-6pt) -- (0,4pt);
    \draw[onearrow](0,4pt) -- (0,14pt);
  }\,
  = (-1)^{d}
  \,\tikz[baseline=(basepoint)]{ 
    \path (0,0) coordinate (basepoint) (0,4pt) node[dot] {};
    \draw[] (-8pt,-6pt) .. controls +(8pt,4pt) and +(8pt,-4pt) .. (0,4pt);
    \draw[] (8pt,-6pt) .. controls +(-8pt,4pt) and +(-8pt,-4pt) .. (0,4pt);
    \draw[onearrow](0,4pt) -- (0,14pt);
  }\,}_{\text{homological degree $-d$}}
  $$
  and relations
  \begin{gather*}
  \,\tikz[baseline=(basepoint)]{ 
    \path (0,3pt) coordinate (basepoint) (0,1pt) node[dot] {} (-8pt,11pt) node[dot] {};
    \draw[onearrow](0,-9pt) -- (0,1pt);
    \draw[onearrow](0,1pt) -- (-8pt,11pt);
    \draw[onearrow](-8pt,11pt) -- (-12pt,21pt);
    \draw[onearrow](-8pt,11pt) -- (0pt,21pt);
    \draw[onearrow](0,1pt) -- (12pt,21pt);
  }\,
  =
  \,\tikz[baseline=(basepoint)]{ 
    \path (0,3pt) coordinate (basepoint) (0,1pt) node[dot] {} (8pt,11pt) node[dot] {};
    \draw[onearrow](0,-9pt) -- (0,1pt);
    \draw[onearrow](0,1pt) -- (8pt,11pt);
    \draw[onearrow](0,1pt) -- (0pt,21pt);
    \draw[onearrow](8pt,11pt) -- (-12pt,21pt);
    \draw[onearrow](8pt,11pt) -- (12pt,21pt);
  }\,
  =
  \,\tikz[baseline=(basepoint)]{ 
    \path (0,3pt) coordinate (basepoint) (0,1pt) node[dot] {} (8pt,11pt) node[dot] {};
    \draw[onearrow](0,-9pt) -- (0,1pt);
    \draw[onearrow](0,1pt) -- (8pt,11pt);
    \draw[onearrow](0,1pt) -- (-12pt,21pt);
    \draw[onearrow](8pt,11pt) -- (0pt,21pt);
    \draw[onearrow](8pt,11pt) -- (12pt,21pt);
  }\,
  ,\quad (-1)^{d}
  \,\tikz[baseline=(basepoint)]{ 
    \path (0,3pt) coordinate (basepoint) (8pt,1pt) node[dot] {} (0,11pt) node[dot] {};
    \draw[onearrow](-12pt,-9pt) -- (0,11pt);
    \draw[onearrow](0pt,-9pt) -- (8pt,1pt);
    \draw[onearrow](12pt,-9pt) -- (8pt,1pt);
    \draw[onearrow](8pt,1pt) -- (0,11pt);
    \draw[onearrow](0,11pt) -- (0,21pt);
  }\,
  =
  \,\tikz[baseline=(basepoint)]{ 
    \path (0,3pt) coordinate (basepoint) (8pt,1pt) node[dot] {} (0,11pt) node[dot] {};
    \draw[onearrow](0pt,-9pt) -- (0,11pt);
    \draw[onearrow](12pt,-9pt) -- (8pt,1pt);
    \draw[onearrow](-12pt,-9pt) -- (8pt,1pt);
    \draw[onearrow](8pt,1pt) -- (0,11pt);
    \draw[onearrow](0,11pt) -- (0,21pt);
  }\,
  =
  \,\tikz[baseline=(basepoint)]{ 
    \path (0,3pt) coordinate (basepoint) (-8pt,1pt) node[dot] {} (0,11pt) node[dot] {};
    \draw[onearrow](12pt,-9pt) -- (0,11pt);
    \draw[onearrow](-12pt,-9pt) -- (-8pt,1pt);
    \draw[onearrow](0pt,-9pt) -- (-8pt,1pt);
    \draw[onearrow](-8pt,1pt) -- (0,11pt);
    \draw[onearrow](0,11pt) -- (0,21pt);
  }\,
  ,\\
  \,\tikz[baseline=(basepoint)]{ 
    \path (0,3pt) coordinate (basepoint) (0,1pt) node[dot] {} (0,11pt) node[dot] {};
    \draw[onearrow](-8pt,-9pt) -- (0,1pt);
    \draw[onearrow](8pt,-9pt) -- (0,1pt);
    \draw[onearrow](0,1pt) -- (0,11pt);
    \draw[onearrow](0,11pt) -- (-8pt,21pt);
    \draw[onearrow](0,11pt) -- (8pt,21pt);
  }\,
  =
  \,\tikz[baseline=(basepoint)]{ 
    \path (0,3pt) coordinate (basepoint) (-8pt,1pt) node[dot] {} (8pt,11pt) node[dot] {};
    \draw[onearrow](-8pt,-9pt) -- (-8pt,1pt);
    \draw[onearrow](8pt,-9pt) -- (8pt,11pt);
    \draw[onearrow](-8pt,1pt) -- (8pt,11pt);
    \draw[onearrow](-8pt,1pt) -- (-8pt,21pt);
    \draw[onearrow](8pt,11pt) -- (8pt,21pt);
  }\,
  =
  \,\tikz[baseline=(basepoint)]{ 
    \path (0,3pt) coordinate (basepoint) (8pt,1pt) node[dot] {} (-8pt,11pt) node[dot] {};
    \draw[onearrow](-8pt,-9pt) -- (-8pt,11pt);
    \draw[onearrow](8pt,-9pt) -- (8pt,1pt);
    \draw[onearrow](8pt,1pt) -- (-8pt,11pt);
    \draw[onearrow](-8pt,11pt) -- (-8pt,21pt);
    \draw[onearrow](8pt,1pt) -- (8pt,21pt);
  }\,
  =
  \,\tikz[baseline=(basepoint)]{ 
    \path (0,3pt) coordinate (basepoint) (0,1pt) node[dot] {} (0,11pt) node[dot] {};
    \draw[onearrow](0,-9pt) -- (0,1pt);
    \draw[onearrow](0,1pt) -- (0,11pt);
    \draw[onearrow](0,11pt) -- (0,21pt);
    \draw[onearrow](-16pt,-9pt) -- (0,11pt);
    \draw[onearrow](0pt,1pt) -- (-16pt,21pt);
  }\,
  =
  \,\tikz[baseline=(basepoint)]{ 
    \path (0,3pt) coordinate (basepoint) (0,1pt) node[dot] {} (0,11pt) node[dot] {};
    \draw[onearrow](0,-9pt) -- (0,1pt);
    \draw[onearrow](0,1pt) -- (0,11pt);
    \draw[onearrow](0,11pt) -- (0,21pt);
    \draw[onearrow](16pt,-9pt) -- (0,11pt);
    \draw[onearrow](0pt,1pt) -- (16pt,21pt);
  }\,.
  \end{gather*}
  The first  generator makes any $\Frob_{d}$-algebra $V$ into a (noncounital)
   cocommutative  coalgebra; the second  makes $V[-d]$ into a (nonunital) commutative  algebra.  Typical examples of $\Frob_{d}$-algebras are the  homology $\H_{\bullet}(M)$ and shifted cohomology $\H^{d-\bullet}(M)$ of any $d$-dimensional possibly-noncompact oriented manifold $M$.
  
  The properad $\invFrob_{d}$ of \define{involutive} Frobenius algebras has additionally the relation   $
  \,\tikz[baseline=(basepoint)]{
    \path (0,0pt) coordinate (basepoint) (0,1pt) node[dot] {} (0,7pt) node[dot] {};
    \draw (0,-4pt) -- (0,0pt); \draw (0,8pt) -- (0,12pt);
    \draw (0,1pt) .. controls +(4pt,1pt) and +(4pt,-1pt) .. (0,7pt);
    \draw (0,1pt) .. controls +(-4pt,1pt) and +(-4pt,-1pt) .. (0,7pt);
  }\,
  = 0$.  When $d$ is odd, $\invFrob_{d} = \Frob_{d}$ as the involutivity relation holds automatically by a symmetry/antisymmetry argument.  
\end{definition}

The properads $\Frob_{d}$ and $\invFrob_{d}$ are naturally genus-graded by declaring that the generators have genus $0$.  When $d$ is even, $\Frob_{d}(\m,\n,\beta) \cong \KK$ for all $\beta\in\NN$ and $\m,\n \neq \emptyset$.  For all $d$, $\invFrob_{d}(\m,\n,0) \cong \KK$ and $\invFrob_{d}(\m,\n,\beta) = 0$ for $\beta > 0$.  The dioperad $\Frob_{d}$ satisfies $\Frob_{d}(\m,\n) \cong \KK$ for all $\m,\n \neq \emptyset$.  Thus $\Frob_{d}$ and $\invFrob_{d}$ are positive and locally finite-dimensional.

\begin{definition}
  The (di/pr)operad $\LB_{d}$ of \define{$d$-shifted Lie bialgebras}
is generated by  a degree-$(d-1)$ \define{bracket} $\mult$ satisfying signed-symmetry and Jacobi identities and a degree-$(-1)$ \define{cobracket} $\comult$ satisfying symmetry and Jacobi and cocycle identities.  (Our grading convention is that a  Lie bialgebra structure in the usual sense on $V$ is the same as an $\LB_{2}$-algebra structure on $V[-1]$.)
   \end{definition}
   
The properad $\LB_{d}$ is genus-graded by assigning $\mult$ and $\comult$ each genus $0$.  Both the dioperad and properad $\LB_{d}$ are positive and locally finite-dimensional.

We will need the following version of the (co)bar construction:

\begin{definition}
 Let $P$ be a positive locally finite-dimensional genus-graded properad (resp.\ positive locally finite-dimensional dioperad), and let $\bar P = P/P(\mathbf 1,\mathbf 1,0)$ (resp.\ $\bar P = P/P(\mathbf 1,\mathbf 1)$).  Denote its graded linear dual by $\bar P^{*}$.   The \define{bar dual} to $P$ is the positive locally finite-dimensional genus-graded (di/pr)operad  $\DD(P) = \Ff(\bar P^{*}[-1])$ freely generated by the shifted $\SS$-bimodule $\bar P^{*}[-1]$, whose differential (in addition to any dg structure on $P$) is the derivation defined on generators to be dual to $\sum_{\Gamma \text{ with two vertices}} \Circ_{\Gamma}$.
\end{definition}

\begin{proposition}[\cite{MR1960128,MR2320654}] \label{prop.ddish}
  Let $P$ be any positive locally finite-dimensional dioperad or genus-graded properad.  Then $\DD(P)$ is cofibrant.  Moreover, $\h P = \DD(\DD(P))$ is a cofibrant replacement of $P$. \qedhere
\end{proposition}
 The universal enveloping properad functor is known not to be exact~\cite[Theorem~47]{MR2560406}, so {a priori} the universal enveloping properad of the cofibrant replacement of a dioperad $P$ may not be a cofibrant replacement of the universal enveloping properad of $P$.  When we need to distinguish, we will use superscripts, e.g.\ $\h^{\di}\Frob_{d}$ or $\h^{\pr}\Frob_{d}$.

\begin{proposition}\label{prop.koszulity}
  There is a canonical cofibrant replacement of dioperads $\DD(\Frob_{d}) \to \LB_{d}$.  There is a canonical cofibrant replacement of properads $\DD(\invFrob_{d}) \to \LB_{d}$.
\end{proposition}

\begin{proof}
  There is a good theory of quadratic and Koszul (di/pr)operads, which we will not review; details are available in~\cite{MR1960128,MR2320654}.  In particular, results therein prove that 
  for all $d$, the dioperads $\Frob_{d}$ and $\LB_{d}$ and the properads $\invFrob_{d}$ and $\LB_{d}$ are Koszul.
  Moreover, the dioperads $\Frob_{d}$ and $\LB_{d}$ are quadratic duals, as are the properads $\invFrob_{d}$ and $\LB_{d}$.  Together these facts imply Proposition~\ref{prop.koszulity}.
\end{proof}

By construction, the dioperad $\DD(\Frob_{d})$ and the properad $\DD(\invFrob_{d})$ have the same presentation, and so the properad $\DD(\invFrob_{d})$ is the universal enveloping properad of the dioperad $\DD(\Frob_{d})$.  We will denote both by the name $\shLB_{d}$.  In detail, $\shLB_{d}$ is has a generating corolla with $m>0$ inputs and $n>0$ outputs in homological degree $\deg\bigl(\tikz[baseline=(basepoint)] {
    \path (0,0) coordinate (basepoint) (0,4pt) node[dot] {} (0,-3pt) node{$\scriptstyle \dots$} (0,11pt) node{$\scriptstyle \dots$};
    \draw[onearrow] (-8pt,-3pt) -- (0,4pt);
    \draw[onearrow] (8pt,-3pt) -- (0,4pt);
    \draw[onearrow] (0,4pt) -- (-8pt,11pt);
    \draw[onearrow] (0,4pt) -- (8pt,11pt);
  }\bigr) = d(m-1)-1$ for each $(m,n)$; such a  generator transforms trivially under permutations of the outgoing edges, and either trivially, if $d$ is even, or by the sign representation, if $d$ is odd, under permutations of the incoming edges.  The differential $\partial\bigl(\tikz[baseline=(basepoint)] {
    \path (0,0) coordinate (basepoint) (0,4pt) node[dot] {} (0,-3pt) node{$\scriptstyle \dots$} (0,11pt) node{$\scriptstyle \dots$};
    \draw[onearrow] (-8pt,-3pt) -- (0,4pt);
    \draw[onearrow] (8pt,-3pt) -- (0,4pt);
    \draw[onearrow] (0,4pt) -- (-8pt,11pt);
    \draw[onearrow] (0,4pt) -- (8pt,11pt);
  }\bigr)$ is a sum over trees with two vertices, each labeled by the corresponding generating corolla.

Suppose that $V$ is an $\shLB_{d}$-algebra.  Think of the completed symmetric algebra $\widehat\Sym(V)$ as the algebra of functions on an infinitesimal manifold (so $V$ is the space of linear functions for a chosen coordinate chart).  Each generating corolla $\tikz[baseline=(basepoint)] {
    \path (0,0) coordinate (basepoint) (0,4pt) node[dot] {} (0,-3pt) node{$\scriptstyle \dots$} (0,11pt) node{$\scriptstyle \dots$};
    \draw[onearrow] (-8pt,-3pt) -- (0,4pt);
    \draw[onearrow] (8pt,-3pt) -- (0,4pt);
    \draw[onearrow] (0,4pt) -- (-8pt,11pt);
    \draw[onearrow] (0,4pt) -- (8pt,11pt);
  }$ with $m$ inputs and $n$ outputs defines a map $V^{\otimes m} \to \Sym^{n}(V)$, which we can extend to map $L_{(m)}^{(n)} : \widehat\Sym(V)^{\otimes m} \to \widehat\Sym(V)$ by declaring that it is a derivation in each variable.  Let $L_{(m)} = \sum_{n} \frac1{n!} L_{(m)}^{(n)}$.  Direct computation proves:

\begin{proposition}
\label{prop.poisd}
  Let $V$ be an $\shLB_{d}$-algebra, and define $m$-linear multiderivations $L_{(m)}
  $ as above.  The $L_{(m)}$s together make $\widehat\Sym(V)[d-1]$ into a flat $\Lie_{\infty}$-algebra.  
  Conversely, any flat $\Lie_{\infty}$-algebra structure on $\widehat\Sym(V)[d-1]$ for which all operations are continuous multiderivations that vanish at the origin makes $V$ into an $\shLB_{d}$-algebra. \qedhere
\end{proposition}

A \define{strict $\Pois_{d}$-algebra} is a dg commutative algebra $A$ along with a biderivation $A\otimes A \to A$ making $A[d-1]$ into a dg Lie algebra~\cite{costellogwilliam}.    Just as the notion of $\Lie_{\infty}$-algebra provides a homotopical weakening of the notion of Lie algebra, Proposition~\ref{prop.poisd} provides a homotopical weakening of strict $\Pois_{d}$ infinitesimal manifolds.  We therefore define:

\begin{definition}
  A \define{semistrict homotopy $\Pois_{d}$ infinitesimal manifold} is an $\shLB_{d}$-algebra.
\end{definition} 

The word ``semistrict'' denotes that we do not weaken the Leibniz rule.  Semistrict homotopy $\Pois_{d}$ manifolds were first introduced, with different degree conventions, in~\cite{CFL}.

The \emph{properad} $\DD(\Frob_{0})$ will play a central role in the quantum version of our story.  It is not $\shLB_{0}$.  Rather:

\begin{proposition}[\cite{MR2591885}]  \label{shbdfthm}
$\DD(\Frob_{0})$ is a positive locally finite-dimensional genus graded properad with generators $\gamma_{m,n,\beta} \in \DD(\Frob_{0})(m,n,\beta)$ for all $(m,n,\beta) \in \NN^{3}$ with $m,n \neq 0$ and $(m,n,\beta) \neq (1,1,0)$.  An action of $\DD(\Frob_{0})$ on $V$ is the same as a continuous $\KK\[\hbar\]$-linear degree-$(-1)$ operation $\Delta = o(1)$ on $\widehat\Sym(V)\[\hbar\] = \widehat\Sym(V\oplus \KK\hbar)$ that vanishes at the origin, annihilates $1$, agrees with an $m$th-order differential operator modulo $\hbar^{m}$, and satisfies the Maurer--Cartan equation $(\partial_{V} + \Delta)^{2} = 0$.  Specifically, one extends the action of $\gamma_{m,n,\beta}$ as an $m$th-order differential operator, and sets $\Delta = \sum_{m,n,\beta} \hbar^{\beta+m-1}\gamma_{m,n,\beta}$.  \qedhere
\end{proposition}
The properad $\Frob_0$ is known to be Koszul~\cite{CMW2014}, implying that the properad $\DD(\Frob_0)$ is a cofibrant replacement of the properad $\invLB_0$ of involutive Lie bialgebras, which is the quotient of $\LB_0$ by the ideal generated by $
  \,\tikz[baseline=(basepoint)]{
    \path (0,0pt) coordinate (basepoint) (0,1pt) node[dot] {} (0,7pt) node[dot] {};
    \draw (0,-4pt) -- (0,0pt); \draw (0,8pt) -- (0,12pt);
    \draw (0,1pt) .. controls +(4pt,1pt) and +(4pt,-1pt) .. (0,7pt);
    \draw (0,1pt) .. controls +(-4pt,1pt) and +(-4pt,-1pt) .. (0,7pt);
  }\,$.
We will not use the Koszulity of the properad $\Frob_{0}$, but we mention it for those readers interested in pursuing the connection with Batalin--Vilkovisky integrals discussed in Remark~\ref{remark.bvconnection}.

The \define{tensor product} of dioperads (resp.\ genus-graded properads) $P$ and $Q$ satisfies $(P\otimes Q)(\m,\n) = P(\m,\n) \otimes Q(\m,\n)$ (resp.\ $(P\otimes Q)(\m,\n,\beta) = P(\m,\n,\beta) \otimes Q(\m,\n,\beta)$).  Any homomorphism $R \to P\otimes Q$ gives a way to turn a tensor product of a $P$-algebra with a $Q$-algebra into an $R$-algebra.

\begin{proposition}
\label{prop.tensorformulae}
  For each $d,d'$, there are homomorphisms:
  \begin{gather*} 
    \label{tensorformulae1} \shLB_{d'-d} \to {\shLB_{d'}} \otimes {\hFrob_{d}}, \\
    \label{tensorformulae2} \hFrob_{d'+d} \to {\hFrob_{d'}}\otimes {\hFrob_{d}}.
  \end{gather*}
  These formulas hold regardless of whether they are interpreted in the dioperadic or properadic sense.
\end{proposition}
\begin{proof}
  Since $\shLB_{d'} \to \LB_{d'}$ and $\hFrob_{d} \to \Frob_{d}$ are cofibrant replacements for all $d,d'$, and since the category of $\SS$-bimodules is semisimple in characteristic $0$, it suffices to witness maps $\LB_{d'-d} \to \LB_{d'}\otimes \Frob_{d}$ and $\Frob_{d'+d} \to \Frob_{d'} \otimes \Frob_{d}$.  On generators these are $\mult \mapsto \mult\otimes\mult$ and $\comult \mapsto \comult \otimes\comult$, and the reader may easily check the relations.
\end{proof}

Moreover, we have the following maps, which are part of a much more general story but suffice for our purposes:

\begin{proposition} \label{prop.tensorformulae5}
  Let $P$ be any positive locally finite-dimensional genus-graded properad.  There is a canonical homomorphism $\hBDF \to P \otimes \DD(P)$ taking the generator $\gamma_{m,n,\beta}$ to the canonical element in $P(m,n,\beta) \otimes P(m,n,\beta)^{*}[-1] \subseteq (P \otimes \DD(P))(m,n,\beta)$.
\end{proposition}
\begin{proof}
  That this map respects the differentials follows from the fact that the differential on $\DD(P)$ encodes composition in $P$.
\end{proof}

For example, consider the case when $P = \shLB_{d}$ for $d$ even.  Then $\shLB_{d}(m,n,\beta)$ has a basis consisting of all dags $\Gamma$ with $m$ inputs, $n$ outputs, each vertex labeled by an ``internal genus,'' and with ``total genus'' (the sum of internal genera and the genus of the graph) $\beta$.  By Proposition~\ref{prop.ddish}, the properad $\DD(\shLB_{d}) = \DD(\DD(\invFrob_{d}))$ is a cofibrant replacement of $\invFrob_{d}$.  It is convenient to describe the basis for $(\shLB_{d})^{*}$ with the same labeled graphs, but with the pairing $\<\Gamma,\Gamma\> = \left|\Aut\Gamma\right|$, where $\Aut\Gamma$ is the group of automorphisms of $\Gamma$.  Then the generator $\gamma_{\m,\n,\beta}$ of $\hBDF$ gets mapped via Proposition~\ref{prop.tensorformulae5} to a sum over ``Feynman diagrams.''

\begin{remark}\label{remark.bvconnection}
Another relation with Feynman diagrams is the following.  The Batalin--Vilkovisky approach to perturbative oscillating integrals of the form $\int f \exp\bigl( \frac i \hbar s\bigr)$ constructs, for each critical point of the ``action'' $s$, a differential on a completed symmetric algebra over $\KK\[\hbar\]$ which is a second-order differential operator, but a derivation modulo $\hbar$ (see for example~\cite{GJF2012}).  A graded commutative algebra over $\KK\[\hbar\]$ with such a differential is called a \define{Beilinson--Drinfeld algebra} in \cite{costellogwilliam}; the name \define{Batalin--Vilkovisky algebra} means something slightly different.

 In particular, any peturbative oscillating integral gives an example of a $\hBDF$-algebra; more generally the BV-BRST formalism constructs a $\hBDF$-algebra (supported in both positive and negative homological degrees) for any gauged perturbative oscillating integral.  In such examples, the differential $\Delta$ is always a second-order operator, and its principal symbol always satisfies a nondegeneracy condition.  Arbitrary $\hBDF$-algebras could be thought of as ``homotopic'' versions of perturbative gauged oscillating integrals ``with parameters''; they can arise in physically meaningful situations when the second-order BV operator requires ``quantum corrections''~\cite{MR1428393}.  The homological perturbation lemma (see e.g.~\cite{Crainic04,MR2762538}; we will review it in the proof of Assertion~\ref{thm.ed}) when applied to $\hBDF$-algebras coming from oscillating integrals produces exactly the usual sum of Feynman diagrams, and for a general $\hBDF$-algebra produces a more complicated diagrammatic sum.

Further relations between properads and Feynman diagrams are discussed in~\cite{MR2600029}.
\end{remark}

\section{\texorpdfstring{An $\hdiFrob_{d}$-algebra}{A homotopy Frobenius algebra} structure on \texorpdfstring{$\Chains_{\bullet}(M)$}{Chains(M)}} \label{section.frobchains}

Suppose $M$ is an oriented $d$-dimensional manifold, which need not be compact.  Both the homology $\H_{\bullet}(M)$ and shifted cohomology $\H^{d-\bullet}(M)$ carry $\Frob_{d}$-algebra structures.  In this section we describe a canonical contractible space of lifts of this structure to the chain level.  For convenience, we assume that $M$ is smooth, work over $\KK = \RR$, and take the complex $\Omega_{\cpt}^{-\bullet}(M)[d]$ of compactly-supported smooth de Rham forms as our model of $\Chains_{\bullet}(M)$, with the projective tensor product $\Chains_{\bullet}(M)^{\otimes n} = \Omega_{\cpt}^{-\bullet}(M^{\times n})[nd]$.  We denote the de Rham differential by $\partial$.  
Working with distributional forms makes no difference, and we moreover
 expect that some version of our arguments would work for piecewise-linear manifolds and cellular chains.

Equip $M$ with a complete Riemannian metric $g$; then $M^{m+n}$ carries an induced metric, and we denote by $g(x,y)$ the distance between points $x,y\in M^{m+n}$.  Given an \define{ultraviolet length scale} $\ell \in \RR_{>0}$, let $B_{\ell,g}(M,m+n) \subseteq M^{m+n}$ denote the open neighborhood
$$ B_{\ell,g}(M,m+n) = \{x\in M^{m+n} \st \exists y\in \diag(M) \text{ with } g(x,y) < \ell\}$$
where $\diag: M \to M^{m+n}$ is the diagonal embedding.  Suppose that $m,n\geq 1$.
Then each smooth de Rham form $f \in \Omega^{-\bullet}(M^{m+n})[dn]$ with support inside $B_{\ell,g}(M,m+n)$ determines, by a pull-push construction, a map
$ \tilde f: \Omega^{d-\bullet}(M)^{\otimes m} \to \Omega^{d-\bullet}(M)^{\otimes N}$ 
which restricts to a map
$ \Chains_{\bullet}(M)^{\otimes m} \to \Chains_{\bullet}(M)^{\otimes N} $
that we also call $\tilde f$.  Such a map is \define{$\ell$-quasilocal with respect to~$g$}.  The collection $\qloc_{\ell,g}(m,n)$ of all $\ell$-quasilocal maps $\tilde f$ is an $\SS$-invariant subcomplex of both $\End(\Chains_{\bullet}(M))(m,n)$ and  $\End(\Omega^{d-\bullet}(M))(m,n)$.  In particular, for $m,n\geq 1$, if $\tilde f\in \qloc_{\ell,g}(m,n)$ is $\partial$-closed, then it descends to both homology and cohomology, defining maps $\H_{\bullet}(f) : \H_{\bullet}(M)^{m}\to \H_{\bullet}(M)^{\otimes n}$ and $\H^{d-\bullet}(f) : \H^{d-\bullet}(M)^{\otimes m} \to \H^{d-\bullet}(M)^{\otimes n}$.

A metric $g'$ on $M$ is \define{finer} than $g$ if $g'(x,y) > g(x,y)$ for all $x,y \in M$.  The usual tubular neighborhood theorems imply that for fixed $\ell,m,n$, for all sufficiently fine $g$ the neighborhood $B_{\ell,g}(M) \subseteq M^{m+n}$ contracts onto $\diag(M)$.  It follows that:
\begin{proposition}
\label{prop.hqloc}
  Fix $\ell,m,n$. For sufficiently fine $g$, $\H_{\bullet}\bigl(\qloc_{\ell,g}(m,n)\bigr) \cong \H^{d-\bullet}(M)[-dm]$. \qedhere
\end{proposition}

Composition of elements in $\End(\Chains_{\bullet}(M))$ corresponds to convolution of their integral kernels.  This does not preserve  length scales $\ell$, but the triangle inequality guarantees that the length scales change in a controlled way:
\begin{proposition}
  Given  $f_{1}\in \qloc_{\ell,g}(m_{1},n_{1})$ and $f_{2} \in \qloc_{\ell,g}(m_{2},n_{2})$, any properadic composition of $f_{1}$ with $f_{2}$ is $\ell'$-quasilocal, where $\ell'$  depends on $\ell,m_{1},m_{2},n_{1},n_{2}$ but not on $f,f'$.\qedhere
\end{proposition}

\begin{corollary} \label{cor.filtered}
  There exist explicit real numbers $\ell(1,m,n) < \ell(2,m,n) < \dots \in \RR_{>0}$ such that $\qloc_{g}(m,n) = \bigcup_{\ell}\qloc_{\ell,g}(m,n)$ has the structure of a {filtered properad}, whose $k$th filtered piece is $\qloc_{g}(m,n)_{\leq k} = \qloc_{\ell(k,m,n),g}(m,n)$.  (By convention, $m,n > 0$, and we freely adjoint a unit in $\qloc(1,1)$ at filtration level $0$, and homotopies relating it to a Thom form.) \qedhere
\end{corollary}

If $X \mono Y$ is a closed embedding of oriented manifolds, a \define{Thom form} in a tubular neighborhood $U \supseteq X$ is a de Rham form on $Y$ supported in $U$ representing the class of $1$ in $\Omega^{\dim X - \bullet}(X) \simeq \{ f\in \Omega^{\dim Y - \bullet}(Y) \st \operatorname{support}(f) \subseteq U\}$.

An abstract induction procedure implies that both the dioperad $\hdiFrob_{d}$ and the properad $\hprFrob_{d}$ have filtrations whose $k$th filtered piece $(\hFrob_{d})_{\leq k}$ is free on finitely many generators, and for each generator $\Gamma$ in the $k$th filtered piece, its derivative $\partial\Gamma$ is in the $(k-1)$th filtered piece.

\begin{definition}
  An action of a filtered (di/pr)operad $P$ on $\Chains_{\bullet}(M)$ is \define{quasilocal with respect to $g$} (abbreviated \define{$g$-quasilocal}) if it factors through a homomorphism $\eta : P \to \qloc_{g}$ of filtered properads.
\end{definition}

  Refinements of $g$ do not change the properad $\qloc_{g}$, although the filtration changes to an equivalent one.
   We will thus sometimes drop the subscript $g$ from the notation, and say that an action is \define{quasilocal} if it is quasilocal with respect to some $g$.  When we want to make the dependence on the manifold $M$ explicit, we will call the properad $\qloc(M)$ (or $\qloc_{g}(M)$ when using a specific metric $g$).

We are now ready to prove the main result of this section:

\begin{theorem}
\label{thm.qlocaction}
  Fix a smooth oriented $d$-dimensional manifold $M$ and a positive integer $k$.  Consider the space of all $g$-quasilocal dioperadic actions $\eta$ of the $k$th filtered piece of $(\hdiFrob_{d})_{\leq k}$ on $\Chains_{\bullet}(M)$ such that all $\partial$-closed lifts of basis elements $\Gamma$ of the dioperad $\Frob_{d}$ get mapped under $\eta: (\hdiFrob_{d})_{\leq k} \to \qloc$ to Thom forms around $\diag(M) \mono M^{m_{\Gamma}+n_{\Gamma}}$.    For all sufficiently fine Riemannian metrics $g$ on $M$, this space is contractible.
\end{theorem}

The condition that the lifts of basis vectors act via Thom forms is exactly the condition necessary to assure that the induced $\Frob_{d}$-actions on (co)homology are the standard ones.

\begin{proof}
  We choose the metric $g$ to be sufficiently fine that for all generators $\Gamma$ of $(\hdiFrob_{d})_{\leq k}$, Proposition~\ref{prop.hqloc} holds for $(m,n,\ell) = \bigl(m_{\Gamma},n_{\Gamma},\ell(m_{\Gamma},n_{\Gamma},k)\bigr)$, where $\ell(m_{\Gamma},n_{\Gamma},k)$ is from Corollary~\ref{cor.filtered}.
  
  By construction, the generators of $\h^{\di}\Frob_{d} = \DD(\shLB_{d})$ are enumerated by directed trees in which every vertex $v$ has $m_{v}\geq 1$ inputs and $n_{v}\geq 1$ outputs, and the vertex with $(m_{v},n_{v}) = (1,1)$ is disallowed.  Let $e_{\Gamma}$ denote the number of internal edges in a tree $\Gamma$; it is called the \define{syzygy degree} of $\Gamma$.  Then a generator $\Gamma$ of $\h^{\di}\Frob_{d}$ with $m_{\Gamma}$ inputs, $n_{\Gamma}$ outputs, and syzygy degree $e_{\Gamma}$ is in homological degree $\deg(\Gamma) = e_{\Gamma} + d(1-m_{\Gamma})$.
  
  The lifts of basis elements of $\Frob_{d}$ are precisely the generators $\Gamma$ with syzygy degree $e_{\Gamma} = 0$.  The conditions of the theorem assert that these are mapped via $\eta : (\hdiFrob_{d})_{\leq k} \to \qloc$ to Thom forms.  The space of choices of Thom forms is contractible.
  
  The action of the remaining generators is determined by obstruction theory.  Let $\Gamma$ be a generator with syzygy degree $e_{\Gamma} = 1$.  Then $\eta(\partial\Gamma)$ is a difference of two Thom forms, and hence vanishes in homology.  Thus $\Gamma$ can be represented.  By Proposition~\ref{prop.hqloc}, the space of choices for $\Gamma$ is contractible.
  
  If $\Gamma$ is a generator with syzygy degree $e_{\Gamma} \geq 2$, then $\eta(\partial\Gamma)$ is closed by induction and in homological degree $e_{\Gamma} + d(1-m_{\Gamma}) > d(1-m_{\Gamma})$.  But Proposition~\ref{prop.hqloc} assures us that $\qloc(m_{\Gamma},n_{\Gamma})$ has no homology above degree $d(1-m_{\Gamma})$.  Thus $\eta(\partial\Gamma)$ is exact, and moreover the space of choices for $\Gamma$ is contractible.
\end{proof}

\begin{corollary} \label{cor.qlocaction}
  There is a canonical contractible space of actions of $\hdiFrob_{d}$ on $\Chains_{\bullet}(M)$.
\end{corollary}
\begin{proof}
  Given a metric $g$, let $Q_{g,k}$ denote the space of actions of $(\hdiFrob_{d})_{\leq k} \to \qloc_{g}$ constructed in Theorem~\ref{thm.qlocaction}.  If $g'$ is finer than $g$, then we have an inclusion $Q_{g',k} \mono Q_{g,k}$.  We also have restrictions $Q_{g,k+1} \to Q_{g,k}$.  For each $k$, choose $g_{k}$ fine enough so that $Q_{g_{k},k}$ is contractible.  We thereby get an inverse system $\dots \to Q_{g_{k+1},k+1}\to Q_{g,k}\to \dots \to Q_{g_{0},0}$ of spaces, each of which is contractible, and therefore its homotopy limit $\hlim Q_{g_{k},k}$ is contractible.  
  The inclusion $\qloc_{g} \mono \End(\Chains_{\bullet}(M))$ provides a map  $\hlim Q_{g_{k},k} \to \hlim \bigl( \hom\bigl( (\hdiFrob_{d})_{\leq k} ,\End(\Chains_{\bullet}(M))\bigr)\bigr) \simeq \hom\bigl( \hcolim (\hdiFrob_{d})_{\leq k} ,\End(\Chains_{\bullet}(M))\bigr)$.  But each $(\h^{\di}\Frob_{d})_{\leq k}$ is fibrant-cofibrant, and so $\hcolim (\hdiFrob_{d})_{\leq k} \simeq \colim (\hdiFrob_{d})_{\leq k} = \hdiFrob_{d}$ is a cofibrant replacement of $\Frob_{d}$.
\end{proof}

  The obstruction theory arguments in the proof of Theorem~\ref{thm.qlocaction} are not enough to also construct an action of the properad $\hprFrob_{d}$: because of the presence of genus, there are many generators whose obstructions may not vanish, or for which the space of choices has nontrivial higher homotopies.  Indeed:

\begin{theorem}\label{thm.no1}
  When $M = \RR$, there does not exist a quasilocal $\hprFrob_{1}$-algebra structure on $\Chains_{\bullet}(\RR)$ extending the canonical $\hdiFrob_{1}$ action.
\end{theorem}

\begin{proof}[Outline of proof]
  See~\cite{chainfrob} for details.  Since $\Frob_{1} = \invFrob_{1}$ is Koszul, one may use the properad $\sh^{\pr}\Frob_{1} = \DD(\LB_{1})$ as its cofibrant replacement.  The obstruction corresponding to the graph 
  \begin{tikzpicture}[baseline=(basepoint)]
    \path (0,12pt) coordinate (basepoint) (0pt,5pt) node[dot] {} (6pt,12pt) node[dot] {} (-4pt,18pt) node[dot] {} (2pt,25pt) node[dot]{};
    \draw (0pt,-2pt) -- (0pt,5pt);
    \draw[] (0pt,5pt) -- (6pt,12pt);
    \draw[] (6pt,12pt) -- (-4pt,18pt);
    \draw[] (0pt,5pt) .. controls +(-6pt,6pt) and +(-6pt,-6pt) .. (-4pt,18pt);
    \draw[] (-4pt,18pt) -- (2pt,25pt);
    \draw[] (6pt,12pt) .. controls +(6pt,6pt) and +(6pt,-6pt) .. (2pt,25pt);
    \draw (2pt,25pt) -- (2pt,32pt);
  \end{tikzpicture}
   is multiplication by $-\frac1{12}$, and in particular is not exact in $\qloc(1,1)$.
\end{proof}

The same argument obstructs $g$-quasilocal $\hprFrob_{1}$ actions on $\Chains_{\bullet}(S^{1})$ when the metric $g$ is sufficiently fine.  On the other hand, the calculations of~\cite{chainfrob} show that the graph 
  \begin{tikzpicture}[baseline=(basepoint)]
    \path (0,12pt) coordinate (basepoint) (0pt,5pt) node[dot] {} (6pt,12pt) node[dot] {} (-4pt,18pt) node[dot] {} (2pt,25pt) node[dot]{};
    \draw (0pt,-2pt) -- (0pt,5pt);
    \draw[] (0pt,5pt) -- (6pt,12pt);
    \draw[] (6pt,12pt) -- (-4pt,18pt);
    \draw[] (0pt,5pt) .. controls +(-6pt,6pt) and +(-6pt,-6pt) .. (-4pt,18pt);
    \draw[] (-4pt,18pt) -- (2pt,25pt);
    \draw[] (6pt,12pt) .. controls +(6pt,6pt) and +(6pt,-6pt) .. (2pt,25pt);
    \draw (2pt,25pt) -- (2pt,32pt);
  \end{tikzpicture}
   is the only obstruction to defining a quasilocal $\DD(\LB_{1})$ action on $\Chains_{\bullet}(\RR)$.  It follows that:

\begin{corollary} \label{cor.1dsuccess}
  Let $\surinvLB_{1}$ denote the quotient of the properad $\LB_{1}$ by the ideal generated by the graph
  \begin{tikzpicture}[baseline=(basepoint)]
    \path (0,12pt) coordinate (basepoint) (0pt,5pt) node[dot] {} (6pt,12pt) node[dot] {} (-4pt,18pt) node[dot] {} (2pt,25pt) node[dot]{};
    \draw (0pt,-2pt) -- (0pt,5pt);
    \draw[] (0pt,5pt) -- (6pt,12pt);
    \draw[] (6pt,12pt) -- (-4pt,18pt);
    \draw[] (0pt,5pt) .. controls +(-6pt,6pt) and +(-6pt,-6pt) .. (-4pt,18pt);
    \draw[] (-4pt,18pt) -- (2pt,25pt);
    \draw[] (6pt,12pt) .. controls +(6pt,6pt) and +(6pt,-6pt) .. (2pt,25pt);
    \draw (2pt,25pt) -- (2pt,32pt);
  \end{tikzpicture}%
  .  (This  is a souped-up version of involutivity, hence the name ``surinvolutive.'')  The subproperad $\DD(\surinvLB_{1})$ of $\sh^{\pr}\Frob_{1} = \DD(\LB_{1})$ does act quasilocally on $\Chains_{\bullet}(\RR)$, extending the homotopy action of the dioperad $\Frob_{1}$, and the space of such actions is contractible. \qedhere
\end{corollary}

\section{Classical and quantum AKSZ theories}\label{section.aksz}

\begin{definition} \label{defn.aksz}
Let $M$ be an oriented $d$-dimensional manifold and $V$ an algebra for $\shLB_{d}$.  By Theorem~\ref{thm.qlocaction}, there is a canonical contractible space of  $\hdiFrob_{d}$-algebra structures on $\Chains_{\bullet}(M)$.  Via Proposition~\ref{prop.tensorformulae}, this gives in turn a canonical contractible space of $\shLB_{0}$-algebra structures on $\Chains_{\bullet}(M)\otimes V$.  The \define{classical Poisson AKSZ theory} with source $M$ and target $V$ is $\Chains_{\bullet}(M)\otimes V$ with any of these equivalent $\shLB_{0}$-algebra structures.
\end{definition}

The physical interpretation of Definition~\ref{defn.aksz} is the following.  
By Proposition~\ref{prop.poisd}, the chain complex $V$ is really the vector space of linear functions on a semistrict homotopy $\Pois_{d}$ infinitesimal manifold $\spec\widehat\Sym(V)$. 
 The chain complex $\Chains_{\bullet}(M) \otimes V$ is the vector space of linear functions on $\Maps\bigl(\T[1]M,\spec\widehat\Sym(V)\bigr)$, which is the derived space of fields $\phi: M \to \spec\widehat\Sym(V)$ satisfying the field equation $\d\phi = 0$.  The $\shLB_{0}$ action on $\Chains_{\bullet}(M) \otimes V$ gives a semistrict homotopy $\Pois_{0}$ algebra structure on the algebra $\widehat\Sym\bigl(\Chains_{\bullet}(M) \otimes V\bigr)$ of all observables.

Costello and Gwilliam have proposed that the quantization problem in quantum field theory is precisely the deformation problem from $\Pois_{0}$-algebras to $\E_{0}$-algebras~\cite{costellogwilliam}.  For semistrict homotopy $\Pois_{0}$ infinitesimal manifolds, the corresponding deformation problem is, via Proposition~\ref{shbdfthm}:
\begin{definition}
  A \define{quantization} of an $\shLB_{0}$-algebra $W$ is an $\shBDF$-algebra structure on $W$ that pulls back to the given $\shLB_{0}$-algebra structure under the canonical map $\shLB_{0} \to \shBDF$.
\end{definition}

For quantizations of classical AKSZ theories, it makes sense moreover to ask that $\shBDF$ acts quasilocally.  {A priori}, the quantization problem might be obstructed.  There are, however, universal constructions of some quantum AKSZ theories:

\begin{definition}\label{defn.quantumaksz}
  Let $P$ be a genus-graded positive locally finite-dimensional properad whose genus-zero part is any cofibrant replacement $\h^{\di}\Frob_{d}$ of the dioperad $\Frob_{d}$.  Then the genus-zero part of its bar dual $\DD(P)$ is a cofibrant replacement $\h^{\di}\LB_{d}$ of $\LB_{d}$, and hence any action of $\DD(P)$ on a chain complex $V$ makes $V$ into an $\shLB_{d}$-algebra.
  
  Let $M$ be a $d$-dimensional oriented manifold, and suppose that the $\h^{\di}\Frob_{d}$ action on $\Chains_{\bullet}(M)$ extends to a quasilocal action of $P$.  Then Proposition~\ref{prop.tensorformulae5} provides a quasilocal action of $\hBDF$ on $\Chains_{\bullet}(M) \otimes V$ extending the $\shLB_{0}$ action for any $\DD(P)$-algebra $V$.  This is a \define{path integral quantization} of the AKSZ theory on $M$ valued in $V$.
\end{definition}

The name is justified by the remarks after the proof of Proposition~\ref{prop.tensorformulae5}.  As an example, Corollary~\ref{cor.1dsuccess} implies that the AKSZ theory on $\RR$ with target any surinvolutive $\LB_{1}$-algebra has a canonical path-integral quantization.

We will conclude this article by giving evidence in support of Conjecture~\ref{conj}, which relates  path integral quantization of general AKSZ theories on $\RR^{d}$ to the formality of the $\E_{d}$ operad.

\begin{definition} 
A \define{formality morphism} of a dg algebraic object $X$ is a homomorphism $f$ from a cofibrant replacement of $X$ to a cofibrant replacement of $\H_{\bullet}(X)$ such that $\H_{\bullet}(f) : \H_{\bullet}(X) \to \H_{\bullet}(\H_{\bullet}(X)) = \H_{\bullet}(X)$ is the identity.
\end{definition}

\begin{conjecture} \label{conj}
  For $d\geq 2$, the space of quasilocal properadic $\h\invFrob_{d}$ actions on $\Chains_{\bullet}(\RR^{d})$ (such that the generators of $\h\Frob_{d}$ that lift basis elements of $\invFrob_{d}$ act by Thom forms) 
  is homotopy equivalent to the space of formality morphisms  of the operad $\E_{d}$ of chains on the space of configurations of points in $\RR^{d}$.
\end{conjecture}

We abuse notation and write $\E_{d}$ for what is normally called $\Chains_{\bullet}(\E_{d})$.  When $d \geq 2$, the operad $\E_{d}$ is known to be formal; since $\H_{\bullet}(\E_{d}) = \Pois_{d}$, formality is equivalent to the existence of a universal quantization procedure from homotopy $\Pois_{d}$ algebras to $\E_{d}$ algebras~\cite{MR1718044,LambrechtsVolic2008}.  When $d = 1$ in Conjecture~\ref{conj}, one should say that the space of $\h\invFrob_{1}$ actions on $\Chains_{\bullet}(\RR)$ is equivalent to the space of universal wheel-free deformation quantization procedures for Poisson algebras; both sides of this equivalence are empty (Theorem~\ref{thm.no1} and~\cite{Dito13,Willwacher2013}).

\begin{remark}
  One may easily compute $\H_{\bullet}(\QLoc(\RR^{d})) = \invFrob_{d}$, with the basis of $\invFrob_{d}$ represented by Thom forms.  Thus a quasilocal properadic $\h\invFrob_{d}$ action on $\Chains_{\bullet}(\RR^{d})$ is the same as a formality morphism for $\QLoc(\RR^{d})$, and Conjecture~\ref{conj} may be rephrased as saying the operad $\Chains_{\bullet}(\E_{d})$ and the properad $\QLoc(\RR^{d})$ have canonically homotopy-equivalent spaces of formality morphisms.
\end{remark}

In our explanation of Conjecture~\ref{conj}, we will use the following model of $\E_{d}$:
\begin{definition}
  We continue to let $\m$ denote a set of size $m \in \NN$.
  Let $\Config_{>1}(\m,\mathbb R^{d})$ denote the manifold of maps $\m \to \RR^{d}$ such that the image of any pair of distinct points in $\m$ are sent to points at distance strictly greater than $1$ for the standard metric on $\RR^{d}$.  The homotopy operad $\E_{d}$ satisfies $\E_{d}(\m) = \Chains_{\bullet}(\Config_{>1}(\m,\RR^{d}))$.
  
  With this presentation, operadic composition is not strictly defined.  Rather, one can choose other models of $\E_{d}$, including the operads of little disks or of little rectangles.  There are  quasiisomorphisms relating our model to these, and a homotopy operad structure can be pulled through such quasiisomorphisms via homotopy transfer theory \cite{MR1361938,Markl1999,vdLaan2003,Wilson2004}.
\end{definition}

The main piece of evidence in favor of Conjecture~\ref{conj} is the following result, which is a quasilocal version of the fact that locally constant factorization algebras on $\RR^{d}$ are the same as $E_{d}$ algebras~\cite[Theorem~5.3.4.10]{DAG}:

\begin{claim}\label{thm.ed}
  Let $V$ be any chain complex with differential $\partial_{V}$ extended to $\widehat\Sym(V)$ as a derivation.
  Any quasilocal action of $\hBDF$ on $\Chains_{\bullet}(\RR^{d}) \otimes V$ induces an $\E_{d}$-algebra structure on $\widehat{\Sym}(V)\[\hbar\] = \widehat\Sym(V\oplus \KK\hbar)$, equipped with a differential of the form $\partial_{V} + o(1)$.
\end{claim}

 We call Assertion~\ref{thm.ed} an ``assertion'' rather than a ``theorem'' because we will not give a complete proof.  We will describe, for each $f\in \E_{d}(\m)$, the corresponding map $\widehat\Sym(V\oplus \KK\hbar)^{\otimes m} \to \widehat\Sym(V\oplus \KK\hbar)$.  More work remains to check compositions, so that our construction really gives a homotopy-operad action.
As in Section~\ref{section.frobchains}, we will work over $\KK = \RR$, but we find it clarifying to distinguish  the coefficient field from the manifold $\RR$.

\begin{proof}[Idea of proof]
  Rather than considering the action of $\qloc$ on the complex $\Chains_{\bullet}(\RR^{d})$ of smooth chains, we will use the action on the complex $\Chains^{\dist}_{\bullet}(\RR^{d})$ of distributional chains.  The elements of $\qloc$ and $\E_{d}$ will remain smooth.
  
  We begin by describing the deformed differential on $\widehat\Sym(V)\[\hbar\]$.
  For $x\in \RR^{d}$, let $\delta_{x} \in \Chains^{\dist}_{0}(\RR^{d})$ denote the delta distribution supported at $x$, thought of as a $0$-chain, and $\iota_{x} : \KK \to \Chains^{\dist}_{\bullet}(\RR^{d})^{\otimes m}$ the map $1 \mapsto (\delta_{x})^{\otimes m}$.  Let $\int: \Chains^{\dist}_{\bullet}(\RR^{d}) \to \KK$ denote the map that vanishes in degree $\bullet \neq 0$ and sends compactly supported distributions to their total volumes, and $p = \int^{\otimes n}: \Chains^{\dist}_{\bullet}(\RR^{d})^{\otimes n} \to \KK$.  Then $p$ and $\iota_{x}$ are quasiisomorphisms such that $p\circ\iota_{x} = \id$.  It follows that we can choose $\mathbb S_{n}$-equivariant homotopies $h_{x} : \Chains^{\dist}_{\bullet}(\RR^{d})^{\otimes n} \to \Chains_{\bullet+1}(\RR^{d})^{\otimes n}$ such that $[\partial,h_{x}] = \partial\circ h_{x} + h_{x}\circ\partial = \id - \iota_{x}\circ p$, where $\partial$ denotes the de Rham differential.  By tensoring and taking $\SS_{n}$-invariants, we abuse notation and let $p$, $\iota_{x}$, and $h_{x}$ also denote the induced deformation retraction between $\widehat\Sym(V)\[\hbar\]$ and $\widehat\Sym(V\otimes\Chains^{\dist}_{\bullet}(\RR^{d}))\[\hbar\]$.
  
  We now use Proposition~\ref{shbdfthm} to deform the linear differential $\partial$ on $\widehat\Sym(V\otimes\Chains^{\dist}_{\bullet}(\RR^{d}))\[\hbar\]$ to $\partial + \Delta$, where $\Delta = o(1)$.  The homological perturbation lemma (see e.g.~\cite{Crainic04,MR2762538}) allows us to deform the whole deformation retraction:
  \begin{gather*}
    \tilde\iota_{x} = (1 - h_{x}\circ \Delta)^{-1}\circ \iota_{x} : \widehat\Sym(V)\[\hbar\] \to \widehat\Sym(V\otimes\Chains^{\dist}_{\bullet}(\RR^{d}))\[\hbar\] \\
    \tilde p_{x} = p \circ (1 - \Delta \circ h_{x})^{-1} : \widehat\Sym(V\otimes\Chains^{\dist}_{\bullet}(\RR^{d}))\[\hbar\] \to \widehat\Sym(V)\[\hbar\] \\
    \tilde\Delta_{x} = p \circ (1 - \Delta \circ h_{x})^{-1}\Delta \circ \iota_{x} : \widehat\Sym(V)\[\hbar\] \to \widehat\Sym(V)\[\hbar\] \\
    \tilde h_{x} = h_{x} (1 - \Delta \circ h_{x})^{-1} : \widehat\Sym(V\otimes\Chains^{\dist}_{\bullet}(\RR^{d}))\[\hbar\]  \to \widehat\Sym(V\otimes\Chains^{\dist}_{\bullet}(\RR^{d}))\[\hbar\] 
  \end{gather*}
  Then $\tilde\iota_{x}$ and $\tilde p_{x}$ are quasiisomorphisms between $\widehat\Sym(V\otimes\Chains^{\dist}_{\bullet}(\RR^{d}))\[\hbar\]$ with the differential $\partial + \Delta$ and $\widehat\Sym(V)\[\hbar\]$ with the differential $\partial_{V} + \tilde\Delta_{x}$; that they are quasiisomoprhisms is witnessed by the equations $\tilde p_{x} \circ \tilde \iota_{x} = \id$ and $[\partial+\Delta, \tilde h_{x}] = \id - \tilde \iota_{x} \circ \tilde p_{x}$.
  
Finally, we decide to give $\widehat\Sym(V)\[\hbar\]$ the differential $\partial_{V} + \tilde \Delta_{0}$.  For future use, we let $\tau_{x} = \tilde\iota_{x} \circ \tilde p_{x} \circ \tilde\iota_{0}$.  It is a chain map from $\bigl(\widehat\Sym(V)\[\hbar\], \partial_{V} + \tilde \Delta_{0}\bigr)$ to $\bigl( \widehat\Sym(V\otimes\Chains^{\dist}_{\bullet}(\RR^{d}))\[\hbar\], \partial + \Delta\bigr)$.


  For $k \in \NN$, we will now describe an $\E_{d}$-algebra structure on the finite-dimensional quotient $\Sym^{<k}(V\oplus \KK\hbar) = \Sym(V\oplus \KK\hbar) / \Sym^{\geq k}(V\oplus \KK\hbar)$ of $\widehat\Sym(V\oplus \KK\hbar)$ with the differential $\partial_{V} + \tilde \Delta_{0}$.  Our construction will be compatible with increasing the value of $k$, and so will induce an $\E_{d}$-algebra structure on the projective limit $\widehat\Sym(V\oplus \KK\hbar)$.
  
  As in Proposition~\ref{shbdfthm}, let $\gamma_{m,n,\beta}$ denote the generator of $\hBDF$ with $m$ inputs, $n$ outputs, and genus $\beta$.  It determines an $m$th-order differential operator on $\widehat\Sym\bigl(\Chains^{\dist}_{\bullet}(\RR^{d})\otimes V \oplus \KK\hbar\bigr)$; by an abuse of notation, we will also call this differential operator $\gamma_{m,n,\beta}$.  As an $m$th-order differential operator, $\gamma_{m,n,\beta}$ lowers degree in $\Chains^{\dist}_{\bullet}(\RR^{d})\otimes V$ by $m$.  But it then raises degree in $\Chains^{\dist}_{\bullet}(\RR^{d})\otimes V$ by $n\geq 1$, and also raises degree in $\hbar$ by $\beta + m - 1$.  All together, we see that $\gamma_{m,n,\beta}$ does not lower the total degree of a monomial in $\Chains^{\dist}_{\bullet}(\RR^{d})\otimes V \oplus \KK\hbar$, and therefore descends to the quotient $\Sym^{<k}\bigl(\Chains^{\dist}_{\bullet}(\RR^{d})\otimes V \oplus \KK\hbar\bigr)$ for any $k$.  Moreover, when $m+n+\beta >k$, $\gamma_{m,n,\beta}$ acts by $0$ on this quotient.  Choose $\ell$ such that all generators $\gamma_{m,n,\beta}$ of $\hBDF$ with $m+n+\beta \leq k$ act $\ell$-quasilocally.  By rescaling the standard metric on $\RR^{d}$ by a factor of $2\ell$, we may suppose in fact that all these generators act $\frac12$-quasilocally.  
   
  Suppose that $f\in \E_{d}(\m)$.  By embedding $\Config(\m,\RR^{d})$ into $\RR^{dm}$, we can think of $f = f(x,\dots,z)$ as a smooth chain in $m$ variables $x,\dots,z$ each ranging over $\RR^{d}$.  We can integrate smooth chains against distributions, and so it makes sense to define:
  $$ \tau_{f} = \int_{x,\dots,z} f(x,\dots,z)\,\tau_{x} \otimes \dots \otimes \tau_{z} : \Sym^{<k}(V\oplus \KK\hbar)^{\otimes m} \to \Sym^{<k}\bigl(V\otimes\Chains^{\dist}_{\bullet}(\RR^{d})\oplus\KK\hbar\bigr)^{\otimes m} $$
  Let $\odot: \Sym^{<k}\bigl(V\otimes\Chains^{\dist}_{\bullet}(\RR^{d})\oplus\KK\hbar\bigr)^{\otimes m} \to \Sym^{<k}\bigl(V\otimes\Chains^{\dist}_{\bullet}(\RR^{d})\oplus\KK\hbar\bigr)$ denote the commutative multiplication, and consider the map $\odot\circ \tau_{f}$.  The map $\odot$ is not a map of chain complexes if $\Sym^{<k}\bigl(V\otimes\Chains^{\dist}_{\bullet}(\RR^{d})\oplus\KK\hbar\bigr)$ is given the differential $\partial+\Delta$.  But by construction, $f$ vanishes whenever any two of its variables get within distance $1$ of each other, whereas the failure of $\odot$ to be a map of chain complexes is supported on those chains that are within distance $\frac12$ of some diagonal.  It follows that
  $ (\partial + \Delta)\circ \odot\circ \tau_{f} - \odot\circ\tau_{f} \circ (\partial_{V} + \tilde\Delta_{0}) = \odot\circ\tau_{\partial f}. $
  
  Finally, we declare that $f$ acts on $\Sym^{<k}(V\oplus \KK\hbar)$ by $\eta(f) = \tilde p_{0}\circ\odot \circ \tau_{f} : \Sym^{<k}(V\oplus \KK\hbar)^{\otimes m} \to \Sym^{<k}(V\oplus \KK\hbar)$.  The above calculations prove that $\eta$ is a chain map from $\E_{d}(\m)$ to $\End\bigl(\Sym^{<k}(V\oplus \KK\hbar)\bigr)(\m,\mathbf 1)$.  In our model of $\E_{d}$, composition is defined only up to a system of homotopies.  It is not too difficult to show that for any composition $g\circ f$ in $\E_{d}$, $\tau_{g\circ f}$ and $\tau_{g}\circ \tau_{f}$ are homotopic.  More generally, we expect that $\eta$ extends to a homomorphism of homotopy operads.
\end{proof}

If the quasilocal $\hBDF$ action quantizes an AKSZ theory valued in an $\shLB_{d}$-algebra $V$, then one can show that, modulo $\hbar$, the differential $\tilde\Delta_{0}$ on $\widehat\Sym(V)$ is precisely the extension as a derivation of the sum of the generators of $\shLB_{d}$ with $m=1$ input.  Moreover:

\begin{proposition}
  Suppose that the classical AKSZ theory on $\RR^{d}$ valued in an $\shLB_{d}$-algebra $V$ admits a path integral quantization.  Then in the corresponding $\E_{d}$-algebra structure on $\widehat\Sym(V)\[\hbar\]$, any cycle $f\in \E_{d}(\mathrm 2)$ representing the fundamental class of the $(d-1)$-sphere $S^{d-1} \overset\sim\mono \E_{d}(\mathrm 2) $ acts, modulo $\hbar^{2}$, by $\hbar$ times the binary bracket on $\widehat\Sym(V)$ determined by the $\shLB_{d}$ structure.
\end{proposition}

\begin{proof}
  As in the proof of Assertion~\ref{thm.ed}, we let $p$ denote the extension as an algebra homomorphism of $\int: V \otimes \Chains^{\dist}_{\bullet}(\RR^{d}) \to V$ to a map $\widehat\Sym(V \otimes \Chains^{\dist}_{\bullet}(\RR^{d})) \to \widehat\Sym(V)$, and $\iota_{x}$ is the extension as an algebra homorphism of the map $\otimes\delta_{x} : V \to V \otimes\Chains^{\dist}_{\bullet}(\RR^{d})$, where $x\in \RR^{d}$.  We let $h_{x}$ denote a homotopy between $\id$ and $\iota_{x}\circ p$.
  
  Modulo $\hbar^{2}$, the differential $\Delta$ has three parts:
  \begin{equation} \label{eqn.delta1} \Delta = \sum_{n} \gamma_{1,n,0} + \hbar \sum_{n} \gamma_{1,n,1} + \hbar\sum_{n} \gamma_{2,n,0} \end{equation}
  As in the proof of Assertion~\ref{thm.ed}, $\gamma_{m,n,\beta}$ denotes the extension as an $m$th order differential operator of the action of the generator of $\hBDF$ with $m$ inputs, $n$ outputs, and genus $\beta$.
  
  The algebra $\widehat\Sym(V \otimes \Chains^{\dist}_{\bullet}(\RR^{d}))$ has a bigrading: in addition to the total homological degree (for which $\Delta$ has degree $-1$), there is also the \define{chain degree}.  
  In a path integral quantization, the first sum in equation~\eqref{eqn.delta1} acts entirely in chain degree $0$; the second is in chain degree $1-d$, and the third is in chain degree $-d$.  Moreover, the first term is closed for the de Rham differential.  After unpacking the action of $f$, we can conclude for chain degree reasons that $f$ acts as:
  $$ p \circ   \hbar\sum_{n} \gamma_{2,n,0}  \circ h_{0} \circ \int f(x,y)\iota_{x}\iota_{y} + O(\hbar^{2}) $$
  
  In $V$, $\gamma_{2,n,0}$ acts as the $n$th Taylor coefficient of the binary bracket.   On the chain side, $p \circ \gamma_{2,n,0}$ acts as multiplication, followed by integration.  The homotopy $h_{0}$ antidifferentiates the $(d-1)$-sphere $f$ into a solid $d$-ball $F$ with boundary $\partial F = f$, and by considering chain degree separately in the variables $x$ and $y$, we see that the two inputs of $\gamma_{2,n,0}$ must not act in the same variable.  Therefore: 
  $$ {\textstyle\int}  \circ \mult \circ \int F(x,y)\iota_{x}\iota_{y} = \int \delta_{x-y} F(x,y) = 1 \vspace*{-21pt} $$
\end{proof}

Similar arguments are expected to work in higher arity, showing that a path integral quantization of an AKSZ theory gives an $\E_{d}$ quantization.  Assuming this, we conclude:

\begin{corollary}\label{cor.final}
  Any quasilocal  $\h\invFrob_{d}$ action on $\Chains_{\bullet}(\RR^{d})$ determines a universal wheel-free $\E_{d}$ quantization of $\Pois_{d}$-algebras, and hence implies the formality of the $\E_{d}$ operad when $d\geq 2$.
\end{corollary}
\begin{proof}
  Say that an $\shLB_{d}$-algebra (resp.\ $\hBDF$-algebra) $V$ is \define{polynomial} if 
  for every $ v\in V^{\otimes m}$, there are only finitely many generators with non-zero action on $ v$.
   Any quasilocal $\h\invFrob_{d} = \DD(\shLB_{d})$-action on $\Chains_{\bullet}(\RR^{d})$ gives a universal quantization of AKSZ theories on $\RR^{d}$ with values in an $\shLB_{d}$-algebra, and if the target is polynomial, then all completed symmetric algebras in the proof of Assertion~\ref{thm.ed} can be replaced by non-completed symmetric algebras.   On the other hand, in characteristic $0$, every $\Pois_{d}$-algebra has a resolution by the symmetric algebra on a polynomial $\shLB_{d}$-algebra.
\end{proof}

This gives an alternate proof of Theorem~\ref{thm.no1} obstructing the existence of a properadic homotopy $\Frob_{1}$ action on $\Chains_{\bullet}(\RR)$, since  a wheel-free universal quantization of polynomial $\Pois_{1}$-algebras does not exist~\cite{Dito13,Willwacher2013}.  Indeed, our obstruction is exactly the same as the obstruction found by Merkulov~\cite{MR2397628}.  

To conclude, let us mention how one could prove Conjecture~\ref{conj}.  Corollary~\ref{cor.final} provides one direction of the equivalence.  In the other direction, the formality of $\E_{d}$ implies, via the factorization algebras of \cite{costellogwilliam,DAG}, that there is a universal quantization of AKSZ theories valued in an $\shLB_{d}$-algebra $V$.  This in turn is equivalent to a  system of quasilocal operations on $\Chains_{\bullet}(\RR^{d}) \otimes V$ that depend in a universal way on the action of $\shLB_{d}$ on $V$.  Unpacking this gives some properad that acts quasilocally on $\Chains_{\bullet}(\RR^{d})$.  We expect that this properad is precisely $\h\invFrob_{d}$.


\end{document}